\documentclass[runningheads]{llncs}

\usepackage{graphicx,amssymb,amsmath}
\usepackage[english]{babel}

\newcommand{\R}{\mathbb{R}}
\newcommand{\N}{\mathbb{N}}
\newcommand{\ch}{\mathit{CH}}
\newcommand{\su}[1]{{\mathrm{succ}(#1)}}
\newcommand{\pr}[1]{{\mathrm{pred}(#1)}}

\newtheorem{observation}{Observation}

\newcommand{\keywords}[1]{\par\addvspace\baselineskip%
  \noindent\keywordname\enspace\ignorespaces#1}

\begin{document}

\mainmatter              %
\title{Finding the $\Theta$-Guarded Region}
\titlerunning{Finding the $\Theta$-Guarded Region}  %

\author{Domagoj Matijevi\'c\inst{1}  \and  Ralf Osbild\inst{2}}
\institute{
Department of Mathematics, J.J. Strossmayer University, Osijek, Croatia
\and
Max-Planck-Institut f\"{u}r Informatik, Saarbr\"{u}cken, Germany\\
\email{\{dmatijev,osbild\}@mpi-inf.mpg.de}
}

\maketitle              %

\begin{abstract}
We are given a finite set of $n$ points (\emph{guards}) $G$ in the plane $\R^2$ and an angle $0\le\Theta\le2\pi$.
A $\Theta$-cone is a cone with apex angle $\Theta$.
We call a $\Theta$-cone \emph{empty (with respect to $G$)}
if it does not contain any point of $G$. 
A point $p\in\R^2$ is called \emph{$\Theta$-guarded}
if every $\Theta$-cone with its apex located at $p$ is non-empty.
Furthermore, the set of all $\Theta$-guarded points
is called the \emph{$\Theta$-guarded region},
or the $\Theta$-region for short.

We present several results on this topic.
The main contribution of our work is to
describe the  $\Theta$-region
with $O(\frac{n}{\Theta})$
circular arcs,
and we give an algorithm to compute it.
We prove a tight $O(n)$ worst-case bound on the complexity
of the $\Theta$-region for $\Theta\ge\frac{\pi}{2}.$
In case $\Theta$ is bounded from below by a positive constant,
we prove an almost linear bound
$O(n^{1+\varepsilon})$ for any $\varepsilon>0$ on the complexity.
Moreover,
we show that there is a sequence of inputs such that the asymptotic bound
on the complexity of the $\Theta$-region is $\Omega(n^2)$.
In addition we point out gaps in the proofs of a recent publication
that claims an $O(n)$ bound on the complexity for any constant angle $\Theta.$
\keywords{$\Theta$-guarded region,
unoriented $\Theta$-maxima,
convex hull generalization,
good $\Theta$-illumination,
$\alpha$-embracing contour.}
\end{abstract}

\section{Introduction}\label{s:intro}
Illumination and guarding problems have been a popular topic of study in mathematics and
computer science for several decades. One instance in this class of problems is the classical one 
posed by Victor Klee~\cite{orourke87}: \emph{How many guards are necessary, and how many are sufficient to
patrol the paintings and works of art in an art gallery with $n$ walls?} 
While this particular problem has been solved shortly after by Chvatal~\cite{ch75}
proving a tight $\lfloor \frac{n}{3}\rfloor$ 
bound, many other variants in this problem class have appeared in the literature, 
see e.g. \cite{u-agip-00} for a general survey on the topic.

In this paper
we consider a guarding problem with a fixed number of guards
which have fixed positions in the plane.
We concentrate on the mathematical description, the complexity,
and the computation of the guarded area.

The model is as follows:
We are given a finite set of points (\emph{guards}) $G$ in the plane $\R^2$.
A $\Theta$-cone is a cone with apex angle $\Theta$. We call a 
$\Theta$-cone \emph{empty (with respect to $G$)},
if it does not contain any point of $G$ in its interior.
A point $p\in\R^2$ is called \emph{$\Theta$-guarded (with respect to $G$)},
if every $\Theta$-cone with apex located at $p$ is non-empty.
The set of all $\Theta$-guarded points is called
the \emph{$\Theta$-guarded region},
or the $\Theta$-region for short. 
We consider $\Theta$-cones as open sets,
hence the $\Theta$-region is an open set, too.
The rationale behind this model is that a point is well-guarded only if it is guarded from all sides.

\subsection{Previous Work}
For a given set $G$ of $n$ points in the plane,
Avis et al.~\cite{avis} were the first to
introduce the notion of \emph{unoriented $\Theta$-maxima.}
They say that some point $g\in G$ is a $\Theta$-maxima if there exists
an empty $\Theta$-cone with apex at $g$.
Hence a point $g$ is $\Theta$-maxima if it is not $\Theta$-guarded
with respect to $G.$
They present an $O(\frac{n}{\Theta}\log n)$
algorithm for computing the unoriented $\Theta$-maximum of the set $G,$
or to put it in other words,
an algorithm to query each point in $G$ if it is $\Theta$-guarded or not.
A slight variation of their algorithm can actually query
any finite point-set $P$ in $O(\frac{n+|P|}{\Theta}\log (n+|P|))$ time
as we show in Lemma~\ref{lm:avis}.
They further show that
the unoriented $\frac{\pi}{2}$-maxima
can be computed in $O(n)$ expected time.

Abellanas et~al.~\cite{ABM07} extent the guarding model
(there it is called \emph{good $\Theta$-illumination})
by a range $r$, i.e.,
a guard $g\in G$ can only guard points inside the circle of radius $r$
that is centered at $g$.
Beneath other results they show how to check if a query point $p$
is $\Theta$-guarded in $O(n)$ time and
output the necessary range and guards as witnesses.

Over years several generalizations of the standard convex hull of
a point set
have been proposed, like the $\alpha$-hull~\cite{eks83},
the $k$-th iterated hull~\cite{c85},
and the related concept of the $k$-hull
($k$-depth contour)~\cite{csy87}.

\subsection*{Our contribution}

After some general observations,
we describe the structure of the boundary of the $\Theta$-region
for different values of $\Theta$ in Section~\ref{s:general}.
There we also give an easy and efficient $O(n\log n)$ time algorithm
to compute the boundary in case $\Theta\geq \pi.$
In the main part of the paper we concentrate on the case $\Theta<\pi$,
since for these angles
the problem becomes much more involved and the boundary 
of the $\Theta$-region more complex to understand.
In Section~\ref{s:boundary}
we show that the boundary of the $\Theta$-region is
contained in an arrangement of
circular arcs.
In Section~\ref{s:upperbounds} we bound this set of arcs by
$O(\frac{n}{\Theta})$.
Note that in our work
$\Theta$ and $n$ are independent parameters.
In particular, asymptotic bounds are stated in $n$ and $\frac{1}{\Theta}.$
For $\Theta\ge \frac{\pi}{2}$ we prove that the complexity
of the $\Theta$-region is $O(n)$. 
If $\Theta>\delta$ for a positive constant $\delta>0$,
we show that the complexity is $O(n^{1+\varepsilon}),$
for any $\varepsilon>0.$
In Section~\ref{s:lowerbound}
we give a generic example
for a $\Theta$-region with complexity $\Omega(n^2)$
where the angle $\Theta$ is of order $\frac{1}{n}.$
In Section~\ref{s:algorithm} we give an algorithm
to compute the $\Theta$-region
in $O(n^{\frac{3}{2}+\xi}/\Theta + \mu \log n)$ time,
for any $\xi>0$, 
where $\mu$ denotes the complexity
of the arrangement of the $O(\frac{n}{\Theta})$ arcs.
Our algorithm is  based on the Partitioning Theorem~\cite{Matousek}
and on the computation of an arrangement of circular arcs. 

\paragraph{Remark:}
Besides our work there is an independent and very recent
publication of Abellanas et al.~\cite{ACM08}. 
There the complexity of the $\Theta$-region
(called the $\alpha$-embracing contour) 
is claimed to be {linear}
for {all} constant $\Theta$, and an algorithm that runs in 
$O(n^2\log n)$ time and $O(n^2)$ space is proposed.
For constant $\Theta < \frac{\pi}{2}$
our $O(n^{1+\varepsilon})$ bound on the complexity of the $\Theta$-region
is weaker than the claimed $O(n)$ bound.
However,
we believe that the main results in \cite{ACM08} are not correct,
and we will point out two technical mistakes
in the Appendix.

\subsection{Remark on Plotted Pictures}\label{s:pictures}
The computer generated pictures
are based on the value of the continuous function
$f:\R^2\setminus G\to(0,2\pi]$
where $f(p)=\max\,\{\Theta: \exists\, \textrm{empty $\Theta$-cone with apex $p$} \}.$
The left picture of Figure~\ref{fig:not_connected}
and the two rightmost pictures in Figure~\ref{fig:squared1}
are generated by plotting a grid point shaded,
iff $f$ has a value below the threshold $\Theta$.
In the right picture of Figure~\ref{fig:not_connected}
we have mapped different intervals of function values in $[0,\pi]$
to different gray scale values to visualize isolines
(along the boundary of the gray scale value) of $f$
in this example.
Although these pictures visualize only function values at grid points,
one can rely on the pictures,
since we deal with cones of a certain angle and
not with arbitrarily thin stripes
that could somehow pass between grid points.

\section{The Shape of the $\Theta$-Guarded Region}\label{s:general}
We start with some observations.
A point $p\in\R^2$ does not belong to
the $\Theta$-region, if there is an empty $\Theta$-cone with apex $p$.
Hence, no point inside an empty cone can belong to the region,
and hence, the region can not contain holes.
A point $p$ lies on the boundary of a $\Theta$-region,
if the closure\footnote{Exceptionally we consider closed cones here.}
of each $\Theta$-cone with apex $p$ is non-empty,
and there is at least
one empty (open) $\Theta$-cone with apex at $p.$
The example in Figure~\ref{fig:not_connected} shows
that the $\Theta$-region is not necessarily connected for $0<\Theta<\pi$.
The shape of the $\Theta$-region is invariant under
translation, rotation, and scaling of $G.$
\begin{figure}[t]
  \centering
  \includegraphics[width=.95\columnwidth]{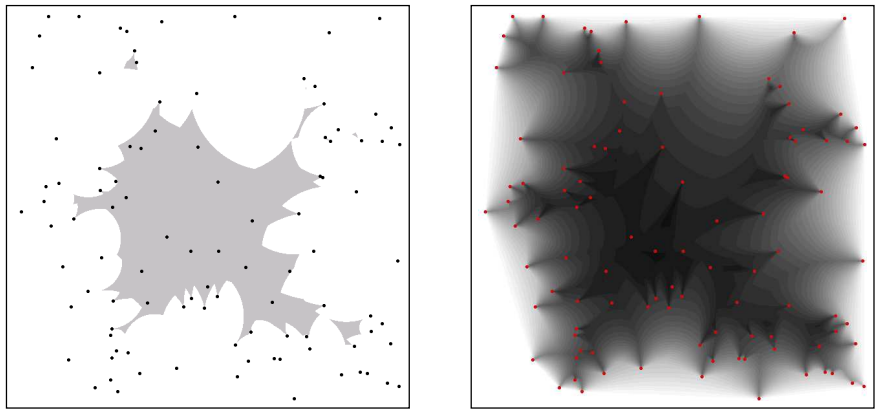}
  \caption{An example with $|G|=50.$
   The $\Theta$-region is not necessarily connected for $0<\Theta<\pi$ (left).
   The isolines of function $f$
   show how components of the $\Theta$-region
   disconnect for decreasing $\Theta$ in this example (right).}
  \label{fig:not_connected}
\end{figure}

The shapes of all $\Theta$-regions can be grouped
according to $\Theta$.
The boundary of the $\pi$-region is just the \emph{convex hull} $\ch(G)$,
because the intersection of all half-planes containing $G$ (convex hull)
is the same as removing every half-plane from $\R^2$
that does not contain any point of $G$ ($\pi$-region).
However, for $0<\Theta<\pi$,
empty (convex) $\Theta$-cones can enter the convex hull through the edges,
while for $\pi<\Theta<2\pi$ the apexes of empty (concave) $\Theta$-cones
do not even have to touch the convex hull (see Figure~\ref{fig:shaperegions}).
\begin{figure}[t]
  \centering
  \includegraphics[width=.8\columnwidth]{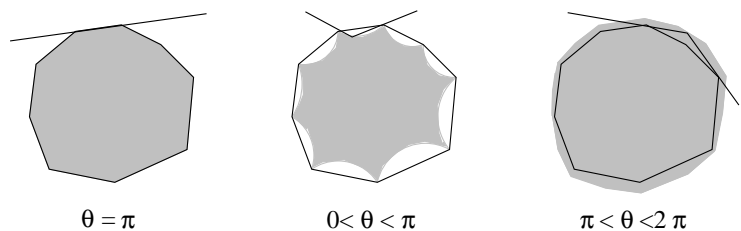}
  \caption{For $\Theta<\pi$ (resp. $\Theta>\pi$)
    the region lies inside (outside) the convex hull $\ch(G)$ and
    the bounding arcs are bend inside (outside) the region.}
  \label{fig:shaperegions}
\end{figure}
Therefore, the $\Theta$-region is connected, if $\Theta\ge\pi$.
Trivially, the $2 \pi$-region is the plane $\R^2$ and the $0$-region is
the empty set.

Before we discuss the $\Theta$-region for $0<\Theta<\pi$
in Sections~\ref{s:boundary}--\ref{s:algorithm},
we discuss the simpler case $\pi<\Theta<2\pi$ below.
Throughout the paper we use
the property about \emph{inscribed angles:}
\emph{Given a circular arc $C_{l,r}$ from $l$ to $r$,
then $\angle lpr = \angle lqr$ holds for all $p,q\in C_{l,r}$.}
We write $C_{l,r}^\alpha$ if the inscribed angle is $\alpha$.
The arc end points are always given in counterclockwise order.

\subsection{Finding the $\Theta$-Guarded Region for $\Theta > \pi$}\label{s:larger}
As already discussed (see Figure~\ref{fig:shaperegions}), 
for $\Theta>\pi$ 
every point in the convex hull interior of $G$ is $\Theta$-guarded. 
Intuitively, the boundary of the $\Theta$-region is drawn
by the apex of an empty
$\Theta$-cone which is rotated around the convex hull $\ch(G)$
such that its rays are always tangent to $\ch(G)$. 
The following algorithm computes the boundary of the $\Theta$-region. 

We first compute the 
clockwise sequence of guards $G'=\{g_1,\ldots, g_k\}$ defining the convex hull
(see for example \cite{ps85}).
Formalizing the intuition given above, we construct an
algorithm that outputs circular arcs defining the boundary of
the $\Theta$-region as follows.
We identify all  pairs $(g_i,g_j) \in G'\times G'$
with $g_i \ne g_j$,
for which there exists an empty $\Theta$-cone
that is tangent to $g_i$ and $g_j,$
and has its apex outside the convex hull.
We say that the apex of the $\Theta$-cone 
can ``see'' the polygonal chain of $\ch(G)$ from $g_i$ to $g_j.$
Such a pair $(g_i,g_j)$ will always have the property,
that the lines supporting the convex hull edges
$(g_j,g_\su{j})$ and $(g_\pr{i},g_i)$
have an angle of intersection not greater than $\Theta,$
and that the lines supporting
$(g_i,g_\su{i})$ and $(g_\pr{j},g_j)$
have an angle of intersection greater than  $\Theta$. 
The sequence of 
all these pairs $(g_i,g_j)$ and the corresponding
circular arcs
$C_{g_j,g_i}^{2\pi-\Theta},$
that are defined by $g_i,g_j,$ and the apex of the $\Theta$-cone
that is tangent to $g_i$ and $g_j,$
can be computed by a cyclic scan over the sequence $G'.$
The arc end points of the $\Theta$-region boundary
can be computed as the intersection points of each circular arc
$C_{g_j,g_i}^{2\pi-\Theta}$
with the supporting lines through
$(g_j,g_\su{j})$ and $(g_\pr{i},g_i).$
Consequently the $\Theta$-region has the same complexity than the convex hull.
The running time of the algorithm is dominated by the convex hull construction
in $O(n\log n)$ time.
We summarize the above in the following Lemma.
\begin{lemma}
The boundary of the  $\Theta$-region for $\Theta > \pi$ can be computed in
time $O(n\log n)$ and its complexity is $|\ch(G)|.$
\end{lemma}

\section{The Boundary of the $\Theta$-Region}\label{s:boundary}
From now on we assume that the angle is $0<\Theta<\pi$.
Here we give a mathematical description of the $\Theta$-region.
First we come back to the inscribed angles and
explain its meaning for our setting.
Let $e=(l,r)\in G\times G$ be any pair of guards.
Then the set of points where we can place the apex
of an empty $\Theta$-cone passing through
the line segment $(l,r)$ in the same direction is
bounded by the circular arc,
incident to $l$ and $r$ having
inscribed angles $\Theta$, and its chord $lr$.
We denote this closed circular segment with
$D_{l,r}^\Theta$ (or $D_e$ for short) and
its bounding circular arc, as above, with $C_{l,r}^\Theta$
(or $C_e$ for short).
Because of the orientation, the circular segment is described uniquely.

The construction of the $\Theta$-region is motivated by the idea
of locally removing sets $T_i$ of unguarded points from
the convex hull $\ch(G)$
such that the remaining part matches the $\Theta$-region
(see Figure~\ref{fig:shaperegions}, middle), i.e.\ we aim for
\begin{eqnarray}\label{for:uniontunnels}
\Theta\textrm{-region} = \ch(G) \setminus \left( \bigcup_{i\in I} T_i \right)
\end{eqnarray}
for specific sets $T_i$.
Next we give the construction for the sets $T_i.$
Consider any empty $\Theta$-cone $c$ that has at least a guard
on each ray (see Figure~\ref{fig:tunnel}, left).
\begin{figure}[t]
  \centering
  \includegraphics[width=.8\columnwidth]{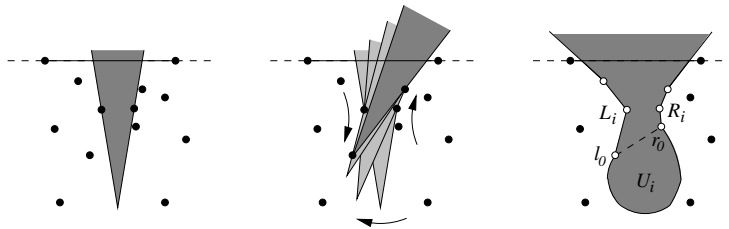}
  \caption{The construction of a tunnel.}
  \label{fig:tunnel}
\end{figure}
First we turn the cone clockwise while pushing the cone
towards the point set, such that it always stays empty but touches a guard
on each boundary (see Figure~\ref{fig:tunnel}, middle).
We end this motion when the apex of the cone reaches the position of a guard, say $l_0.$
Afterwards we start again with cone $c$, i.e.\ in the original position,
and rotate the cone in a similar way counterclockwise
until the apex reaches the position of another guard, say $r_0.$
We extend our notions.
With $L_i$ (resp.~$R_i$) we denote the set of guards
that are incident to the left (resp.~right) ray of a cone
during the construction
(the white points in Figure~\ref{fig:tunnel}, right).
We call the closure of the union of all cones,
which are used during the construction,
the \emph{tunnel $T_i$ with respect to $L_i$ and $R_i$},
or tunnel for short
(shaded region in Figure~\ref{fig:tunnel}, right).
Note that the index set $I$ in Formula~(\ref{for:uniontunnels})
enumerates over all tunnels.

Note that Formula~(\ref{for:uniontunnels})
describes the $\Theta$-region of $G,$
because each empty $\Theta$-cone that intersects $\ch(G)$
lies in at least one tunnel $T_i$:
Let $c$ be such a cone.
We can identify a tunnel by moving $c$
in the direction of its medial axis
until one of its rays is tangent to a guard.
Then we let the empty cone slide along that point without rotation
until the second ray is also tangent to a guard.
According to our construction there is a tunnel that contains this cone
and hence the cone $c$ in its original position.

No point in $T_i$ is $\Theta$-guarded,
but only its boundary can contribute to the boundary of the $\Theta$-region.
First we consider its straight-line boundaries.
Since $\Theta<\pi$,
each point of a straight-line boundary can be crossed infinitesimally
by an empty $\Theta$-cone.
That means,
there are open neighborhoods of unguarded points around
each point of a straight-line boundary,
and hence they can not contribute to the
$\Theta$-region boundary.
Points beyond the straight-line boundaries belong to different tunnels
and will be processed independent from $T_i.$

Therefore we only have to consider the curved boundary of $T_i.$
Observe that during the construction,
the apex of the rotating cone is drawing a \emph{sequence of circular arcs} 
between $l_0$ and $r_0$ which we will formalize next.
We define the set
  \begin{eqnarray}\label{bigintersection}
  U_i \,:= \bigcap_{(l,r)\in L_i\times R_i} D_{l,r}^\Theta
  \end{eqnarray}
as the intersection of all circular segments
for guard pairs in $L_i\times R_i.$
In the following Lemma we state
that we can derive the curved boundary of $T_i$
from these circular segments.
Let $h_i$ be the closed half plane
which is bounded by the line through $l_0$ and $r_0$
and contains the sequence of arcs.

\begin{lemma}\label{lem:tunnel}
  Let $T_i,$ $U_i,$ and $h_i$ be as defined above.
  Then $T_i \cap h_i = U_i.$
\end{lemma}
\begin{proof}
{\bf (Superset.)}
Let $p\in U_i.$
Assume there is no empty cone with apex $p$ through tunnel $T_i.$
This means that there is at least a pair $(l,r)\in L_i\times R_i$
with the property that $\angle l p r<\Theta$.
Hence $p\not\in D_{l,r}^\Theta$ which is a contradiction.
{\bf (Subset.)} Let $p$ be the apex of an empty $\Theta$-cone
through tunnel $T_i.$
This means that $\angle l p r\ge\Theta$ for all $(l,r)\in L_i\times R_i,$
and hence $p$ lies in all corresponding circular segments $D_{l,r}^\Theta$.
\qed
\end{proof}
It follows that the $\Theta$-region boundary is
contained in the curved boundary of the union of the sets $U_i,$ i.e.\
\begin{eqnarray}\label{f:complexity}
  \partial\Theta\textrm{-region}
  \;\subseteq\;
  \partial
  \bigcup_{i\in I} U_i
  \; = \;
  \partial
  \bigcup_{i\in I}
  \left( \bigcap_{(l,r)\in L_i\times R_i} D_{l,r}^\Theta \right),
\end{eqnarray}
where $i$ enumerates over all tunnels.
We observe
that the intersections of
$|L_i|\cdot |R_i|$ circular segments $D_{l,r}$
in Formulae~(\ref{bigintersection})
and (\ref{f:complexity})
are too pessimistic.
During the construction of a tunnel
we collect all guard pairs $(l,r)\subset L_i\times R_i,$
that are incident to the rotating cone simultaneously, in $E_i$.
Since the touching point of $L_i$ (resp.~$R_i$) can only
change in one direction to its neighbor in the sequence of $L_i$ (resp.~$R_i$),
that leads to a set $E_i$ of size $|L_i|+|R_i|-1$.
Therefore we may reduce the intersection of the circular segments
in Formula (\ref{bigintersection}) to
\begin{eqnarray*}\label{f:ui-reduced}
U_i:= \bigcap_{(l,r)\in E_i} D_{l,r} \cap h_i
\end{eqnarray*}
for which Lemma~\ref{lem:tunnel} is still valid.
With ${\cal C}$ we denote the set of all circular arcs
that appear in the boundary of a set $U_i$.

\section{Upper Bounds on the Worst-Case Complexity}\label{s:upperbounds}
Now we discuss the worst-case complexity of the $\Theta$-region
and state the asymptotic bounds in the number $n$ of guards and
the reciprocal value of the angle, i.e.\ $\frac{1}{\Theta}.$
During the analysis of the complexity,
we distinguish cases according to the value of $\Theta.$
We already know that the $0$-region is the empty set.
Since $G$ is a discrete set the $\Theta$-region is also the empty set
for values close to $0.$
\begin{lemma}\label{l:minimaltheta}
The $\Theta$-region for $\Theta\le\frac{2\pi}{n}$ is the empty set.
\end{lemma}
\begin{proof}
Consider the $n$ rays emanating from a point $p\in\R^2\setminus G$
through the guards in $G.$
Then the rays form at least one empty cone with angle of at least
$\frac{2\pi}{n}$ which contains an empty $\Theta$-cone.
Hence $p$ is unguarded.
We can argue similarly for the guards $p\in G.$\qed
\end{proof}
According to the right term in Formula~(\ref{f:complexity})
the complexity of the $\Theta$-region is hidden in
an arrangement of circular arcs.
Since there are at most $O(n^2)$ different circular arcs,
two for each guard pair,
the complexity of the $\Theta$-region is trivially $O(n^4).$
Now we show that the set $\mathcal{C}$ of circular arcs
is of $O(\frac{n}{\Theta})$ size.
Hence the complexity of the $\Theta$-region is
$O(\frac{n^2}{\Theta^2}).$

\begin{theorem}\label{tm:size}
The set $\mathcal{C}$ of circular arcs,
which defines the boundary of the $\Theta$-region, is of
$O(\frac{n}{\Theta})$ size.
\end{theorem}

\begin{proof}
Instead of counting the arcs directly
we count their end points.
Let $p$ be an arc end point of a tunnel
as shown in Figure~\ref{fig:algo}, left.
\begin{figure}[t]
  \centering
  \includegraphics[width=.8\columnwidth]{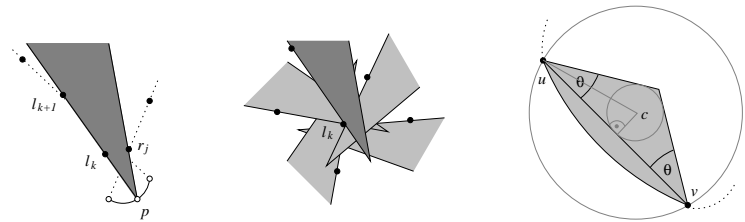}
  \caption{An end point $p$ of a circular arc
           in the boundary of a tunnel (left).
	   A situation that is described in the proof of Theorem~\ref{tm:size}
	   (middle).
	   The auxiliary construction that is described in the proof
	   of Theorem~\ref{t:almostlinear} (right).}
  \label{fig:algo}
\end{figure}
In this position
a ray of the rotating cone is incident to two guards at once.
We assume without loss of generality that two guards lie on the left ray.
We focus on the guard $l_k$ that is closer to the apex and
count how often a guard can be in this situation.
Clearly the number is bounded by $n-1$
because there are no more other guards.
On the other hand we observe
that the empty $\Theta$-cones in this situation
can not intersect each other beyond the second guard on the left ray
(see Figure~\ref{fig:algo}, middle).
Hence there can be at most $\lfloor \frac{2\pi}{\Theta} \rfloor$
different such cones.
With the same argumentation for the right ray
we bound the number of arc end points per guard by
$2\lfloor \frac{2\pi}{\Theta} \rfloor$
and hence the total number of arc end points by
$O(\frac{n}{\Theta}).$\qed
\end{proof}
From the last Theorem we can derive, that
  if the angle $\Theta>\delta$
  is bounded by a constant $\delta>0,$
  then the number of arcs in ${\cal C}$ is $O(n)$
  and the complexity of the $\Theta$-region is $O(n^2).$
With an auxiliary construction we can even further improve this result.
\begin{theorem}\label{t:almostlinear}
  If the angle $\Theta>\delta$
  is bounded by a constant $\delta>0,$
  then the complexity of the $\Theta$-region
  is $O(n^{1+\epsilon})$, for any $\epsilon > 0$.
\end{theorem}
\begin{proof}
We make use of the following construction.
Let $a\in{\cal C}$ be an arc in the boundary of tunnel $T_i,$
and let $u$ and $v$ be the end points of $a.$
The line segment $(u,v)$ and $a$ are the boundary
of a circular segment, say $d_a.$
Now we clue a triangle $t_a$ at the edge $(u,v)$ of $d_a,$
which has an angle of $\min\{\Theta,\frac{\pi}{4}\}$ at $u$ and $v,$
and denote this new object with $F_a:= d_a \cup t_a$
(see Figure~\ref{fig:algo}, right).
We state that the triangle $t_a$ is a subset of $T_i$:
Assume there is a guard $g\in t_a.$
Then the angles $\angle guv$ and $\angle gvu$ are smaller than $\Theta.$
Hence two empty $\Theta$-cones with apexes $u$ and $v$
would belong to different tunnels
what is a contradiction to $a\subset\partial T_i.$
Furthermore, because of the angle at $u$ and $v$
the set of empty $\Theta$-cones with apexes at points in $a$
have to cover $t_a,$
what completes the proof of the statement.

We repeat the above construction for each arc $a\in{\cal C}$
and collect the new objects $F_a$ in the set ${\cal F}.$
We repeat the definition of $\alpha$-fatness from Efrat~et al.~\cite{efrat97}:
\emph{An object $F$ is \emph{$\alpha$-fat} for some fixed $\alpha>1,$
if there exist two concentric disks ${\cal D}\subseteq F \subseteq {\cal D'}$
such that the ratio $\frac{\rho'}{\rho}$
between the radii of ${\cal D'}$ and ${\cal D}$ is at most $\alpha.$}
We state that there is an $\alpha>1$ such that
the objects $F_a\in{\cal F}$ are $\alpha$-fat:
The worst-case scenario occurs when the arc $a$ is almost a straight-line.
Hence we concentrate on the proof that the triangle $t_a$ is $\alpha$-fat
(see again Figure~\ref{fig:algo}, right).
Remember that the angle at $u$ is $\min\{\Theta,\frac{\pi}{4}\}.$
Then the ratio between the radii of the circumcircle and
the inscribed circle is
\begin{eqnarray*}
\frac{\rho'}{\rho} =
\frac{1}{\sin(\frac{1}{2}\cdot\min\{{\Theta},\frac{\pi}{4}\})}
\le
\frac{1}{\sin(\frac{\delta}{2})}
=: \alpha,
\end{eqnarray*}
which is a constant.
Without loss of generality we assume $\delta\le\frac{\pi}{4}.$

The main Theorem in
Efrat~et al.~\cite{efrat97} states,
that the combinatorial complexity of the union of a collection ${\cal F}$
of $\alpha$-fat objects,
whose boundary intersect pairwise in at most $s$ points,
is $O(|{\cal F}|^{1+\varepsilon})$,
for any $\varepsilon>0,$
where the constant of proportionality depends on
$\varepsilon,$ $\alpha,$ and $s.$

It is already shown that $\alpha$ is a constant
and that the objects in ${\cal F}$ are $\alpha$-fat.
The boundary of each convex object $F_a\in{\cal F}$
has always three edges:
two line segments and a circular arc.
Therefore the boundary of each pair of objects in ${\cal F}$
intersect in at most $s=10$ points.
As we said above
$|{\cal C}|\in O(n),$ and hence $|{\cal F}|\in O(n),$
because $\Theta$ is bounded from below by a constant $\delta$.
Therefore the construction fulfills all preconditions
to apply the Theorem of Efrat et al.\ which completes the proof.
\qed
\end{proof}
Now we show that
the complexity of the $\Theta$-guarded region
is linear for angles $\Theta$ at least $\frac{\pi}{2}.$
\begin{theorem}\label{piovertwo}
 The  complexity of the $\Theta$-region is $O(n)$ for
 $\frac{\pi}{2}\le\Theta<\pi$.
\end{theorem}
\begin{proof}
Let ${\cal J}$ be a set of $m$ Jordan curves, i.e.\ simply-closed curves.
Kedem et al.~\cite{kedem90} proved that if any two 
curves in ${\cal J}$ intersect in at most two points
then the complexity of their union is $O(m)$. 

For each set $U_i$ we define a Jordan curve $J_i$.
Let $J_i$ be the curved boundary of $U_i$ from $l_0$ to $r_0$
connected with an auxiliary half circle $C_{r_0,l_0}^\frac{\pi}{2}$, i.e.\
\begin{eqnarray*}
  J_i \,:= \partial
     \left( \bigcap_{(l,r)\in E_i} D_{l,r}^\Theta \cap h_i \right)
     \cup C_{r_0,l_0}^{\frac{\pi}{2}}.
\end{eqnarray*}
Note that the auxiliary half-circle lies in $T_i$
because $\Theta$ is obtuse.
Note further that $J_i$ is the boundary of a convex region
and that $J_i$ lies inside the circle that supports $C_{r_0,l_0}^\frac{\pi}{2}$.
We repeat this construction for all tunnels $T_i$, with $i\in I$,
and collect all curves $J_i$ in ${\cal J}$.
We state that any two curves in ${\cal J}$ intersect at most twice.
Now assume that there are two curves $J_i, J_j \in {\cal J}$
which intersect in more than two points.
We distinguish the following cases as they are shown in Figure~\ref{fig:cases}. 
\begin{figure}[t]
  \centering
  \includegraphics[width=.8\columnwidth]{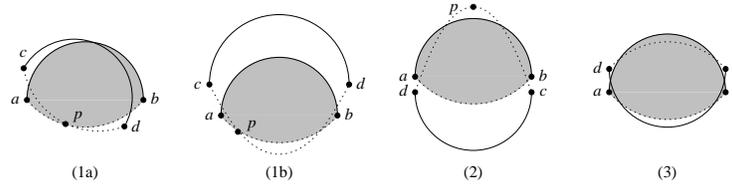}
  \caption{(Proof of Theorem~\ref{piovertwo}.)
    The cases that could imply more
    than just two intersection points between the curves $J_i$ and $J_j$.
    The dotted parts are the sequence of 
    arcs of $U_{i}$ and $U_{j}$ from $a$ to $b$ and from $c$ to $d$,
    respectively,
    and the solid arcs are the auxiliary half-circles.}
  \label{fig:cases}
\end{figure}
Let $A_{i}$ (resp. $A_{j}$) denote the sequence of circular arcs
of $U_i$ (resp. $U_j$). 

In Case~1,
the sequence of circular arcs $A_{i}$ and $A_{j}$ intersect in point $p$
as is shown in Pictures~(1a) or (1b).
That means that for each tunnel $T_i$ and $T_j$
there exists an empty $\Theta$-cone with apex in $p$.
Therefore the angle $\angle apb$ has to be at least $2\Theta$
which is at least $\pi$.
That is a geometrical contradiction.

In Case~2,
we consider a point $p$ that lies on the sequence of arcs $A_{j}$ 
outside $J_i$ as shown in Picture~(2). The angle $\angle cpd$
is at least $\Theta$.
By construction the angle $\angle bpa$
is larger than $\angle cpd$ and hence $\angle apb > \Theta$.
It is a contradiction that $p$ does not lie inside $J_i$.

In Case~3, we consider an empty $\Theta$-cone with apex $c$
through tunnel $T_j$.
Assume this cone passes between $a$ and $b$.
Then the angle $\angle bca$ is at least $\Theta$
and hence $c$ has to lie inside $J_i$. This is a contradiction.
In case that the empty $\Theta$-cone $c$ does not pass between $a$ and $b$,
but $b$ and $c$, or $a$ and $d$,
similar geometric contradictions can be shown.

Other cases are excluded since no guard can lie inside $J_i$ or $J_j$.
This completes the proof.\qed
\end{proof}

\section{Lower Bound on the Worst-Case Complexity}\label{s:lowerbound}
In the following we show that there is a sequence of inputs
such that the asymptotic bound on the complexity of their
$\Theta$-guarded region is $\Omega(n^2)$.
For this purpose
we give a generic construction for point sets $G_i$ with $n_i$ guards
and angles $\Theta_i$ for all $i\in\N$,
such that the complexity of the $\Theta_i$-region of the point set $G_i$
is lower bounded by $c\cdot n_i^2$
for some constant $c$ and $\lim_{i\to\infty}{n_i}=\infty.$
In fact $n_i$ is a linear function in $i,$
and $\Theta_i$ is of order $\frac{1}{i}.$
Therefore the complexity bound
can also be interpreted as $\Omega(\frac{n}{\Theta})$.

First we motivate the construction for a given $i\in\N.$
To achieve the desired complexity,
we construct the point set $G_i$ in such a way that
the $\Theta_i$-region is fragmented into $c\cdot n_i^2$ connected components,
each of constant complexity.
Figure~\ref{fig:squared1} illustrates the idea
of the construction.
\begin{figure}[t]
  \centering
  \includegraphics[width=.95\columnwidth]{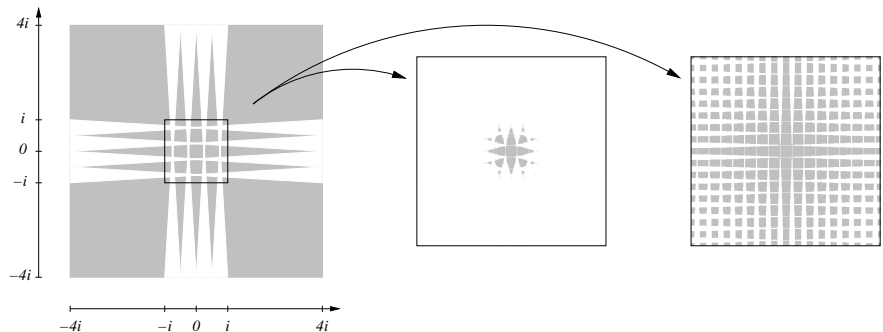}
  \caption{The location of the fragmented $\Theta_i$-region
    is located at the centre of the convex hull (left, $i=2$).
    Some connected components of the center are erased by
    unwanted tunnels (middle, $i=8$).
    All connected components of the center are protected
    against unwanted tunnels (right, $i=8$).}
  \label{fig:squared1}
\end{figure}
The area, where the $\Theta_i$-region is highly fragmented
is at the center of the convex hull.
The decomposition is forced by long, thin tunnels
that enter this area `axis parallel'
from above, below, left, and right;
more precisely
the medial axis of the cones that enter these tunnels deepest
are parallel to the principal axes.
In the \emph{first step} of the construction
we determine the %
tunnels that force the fragmentation and
implicitly determine the area
(bounded by the box in Figure~\ref{fig:squared1})
that contains these %
connected component.
Unfortunately the same guards, that define these tunnels,
define an even larger number of unwanted tunnels which can
enter this area as well.
Therefore we have to place additional guards
in the \emph{second step} with the intention to prevent unwanted tunnels
from entering this area, because they could erase some of
the connected components, hence reducing the total complexity
(see Figure~\ref{fig:squared1}, middle and right).
We show how to place a linear number of additional guards as obstacles
in the plane to keep out a squared number of tunnels from this area.
We note that
because of the construction in the second step the convex hull can be huge
compared to the box in which we count the connected components.
For simplicity reasons we disregard the shape of the $\Theta_i$-region
outside this box.
The construction details are given below.

\subsection{First step: To determine the \emph{wanted} tunnels.}
We denote the square of edge length $2i$ that
is centered at the origin and
is oriented parallel to the principal axes with $B_i.$
In this step guards are placed on the boundary
of the boxes $B_{4i}$ and $B_{2i}.$
The area, in which we will count the connected components, is $B_{i}$
(see Figure~\ref{fig:squared3}, left).
\begin{figure}[t]
  \centering
  \includegraphics[width=.85\columnwidth]{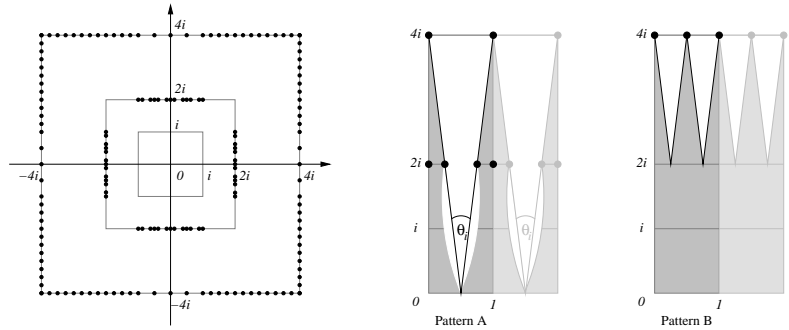}
  \caption{Placement of the guards in the first step for $i=2$ (left).
    Guard patterns A and B (right).}
  \label{fig:squared3}
\end{figure}
The entire construction is symmetric to the origin
as well as to the principal axes.
For this reason we only give the construction for
the upper half of box $B_{4i}$;
the constructions for the lower, left, and right half
of this box are done analogously.

Now we introduce the guard patterns A and B
(see Figure~\ref{fig:squared3}, right),
that define two ways to place guards
inside a cell of width 1 and height $4i$,
which we will use later on to stamp the upper half of the box with.
First we define $\Theta_i$ as the
angle\footnote{Note that $\Theta$ and $i$ are not independent because
$\Theta=\arctan(\frac{1}{8i}) \le \frac{1}{8i}.$}
between the rays emanating from
$(\frac{1}{2},0)$ through the upper corners $(0,4i)$ and $(1,4i).$
To get guard pattern A
we place four guards on the boundary of this cone:
two with $y={2i}$ and two with $y=4i.$
These four guards define a wanted tunnel which is thin in the sense that
the boundary of the tunnel stays in the box of width 1
for values $0\le y\le i$;
remember that we only care for the interior of $B_i.$
For technical reasons we add guards at $(0,2i)$ and $(1,2i)$
to avoid unwanted tunnels between neighboring guard patterns A.
In case we do not need a tunnel inside the cell we use pattern B:
three guards that are placed equidistant on the top edge of the pattern
make it impossible for any cone to enter this box from above deeper
than $y=2i.$
Next we subdivide the upper half of the box $B_{4i}$ in $8i$ cells
of width 1. The medial quarter is stamped with pattern A,
the remaining cells are stamped with pattern B
(see Figure~\ref{fig:squared2}).
\begin{figure}[t]
  \centering
  \includegraphics[width=.95\columnwidth]{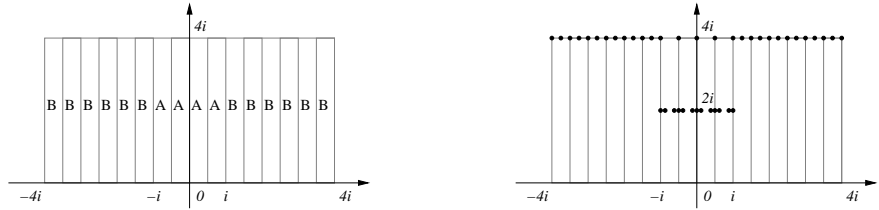}
  \caption{Subdivision of the upper half of $B_{4i}$ in $8i$ cells (left).
     The guard pattern for the upper half (right).}
  \label{fig:squared2}
\end{figure}

This way we can guarantee $2i$ wanted tunnels from above
which intersect $B_i$ and touch the $x$-axis.
After repeating this construction
for the lower, left, and right half of $B_{4i}$,
tunnels from above and below touch at the $x$-axis
as well as tunnels from the left and right touch at the $y$-axis.
This follows immediately from the symmetric construction.
Removing these tunnels from the box $B_i$,
yield to a fragmentation into $(2i+1)^2$ connected components.

Finally we remark that we can save up to 2 guards per stamped pattern
if guards overlap with the neighboring pattern.
Hence the number of guards,
that are placed in the entire first step, sum up to $80i+4.$

\subsection{Second step: To exclude the \emph{unwanted} tunnels.}
Again we start with the construction for the upper half of $B_{4i}.$
Figure~\ref{fig:squared8} shows the situation of the $2i$
neighboring guard patterns A.
We denote the guard pairs on the line $y=4i$ with $P_1,\dots, P_{2i}$
and the guard pairs on the line $y=2i$ with $Q_1, \ldots, Q_{2i}.$
\begin{figure}[t]
  \centering
  \includegraphics[width=.7\columnwidth]{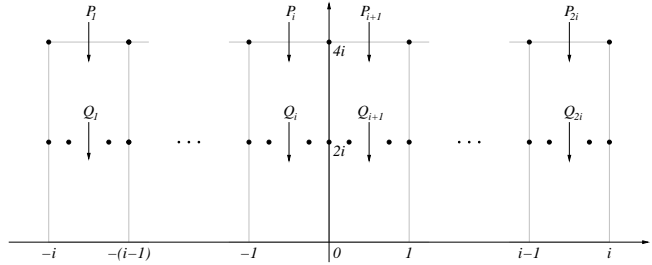}
  \caption{All possible guard pairs through which
    empty cones can reach $B_i$ from above.}
  \label{fig:squared8}
\end{figure}
An empty cone that enters $B_i$ from above therefore has to pass
through $P_k$ and $Q_\ell$ for some $k,\ell\in\{1,\dots,2i\}.$
If $k=\ell,$ the tunnel is wanted.
So the task is to hinder all
cones through tunnels with $k\neq\ell$
to intersect the box $B_i.$
We introduce a new notion to reformulate the problem.
\begin{definition}
Let $t$ be a tunnel through $P_k$ and $Q_\ell.$
We say that an empty cone enters tunnel $t$ \emph{the deepest}
if the $y$-value of its apex is minimal among all empty cones in $t$.
\end{definition}
Note that the deepest cone in a tunnel is unique and
is tangent to at least one guard on each ray.
We define the \emph{slope of a cone} as the slope of its medial axis.
Because of the regular structure of the cells
with guard pattern A we can make the following Observation.
Informally it states that the slope of a deepest cone
through $P_j$ and $Q_{j+h}$
is independent of $j$ and implicitly given by $h.$
\begin{observation}\label{obs:deepest}
  Let $h\in\{0,\dots,2i-1\}.$
  Let $c_j$ be the deepest cone through $P_j$ and $Q_{j+h}$
  and let $d_j$ be the deepest cone through $P_{j+h}$ and $Q_{j}$
  for all $j\in\{1,\dots,2i-h\}.$
  Then all cones $c_j$ have the same slope,
  and all cones $d_j$ have the same slope.
\end{observation}
W.l.o.g.~we concentrate on the cones $c_j$,
but we can argue in the same way for the cones $d_j.$
We derive from Observation~\ref{obs:deepest} that for fixed $h$
the intersection of all deepest cones $c_j$
is again a cone with the same slope
(see Figure~\ref{fig:squared6}, left).
\begin{figure}[t]
  \centering
  \includegraphics[width=.9\columnwidth]{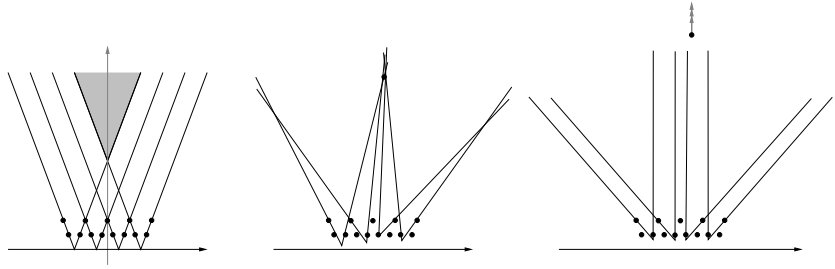}
  \caption{The deepest cones $c_j$ used in Observation~\ref{obs:deepest}.
    Here the vertical arrow marks the direction of their medial axis (left).
    The deepest cones that are tangent to the new vertex (middle).
    The deepest cones, in the case that the new guard is pulled far out
    in the direction of the former medial axis (right).}
  \label{fig:squared6}
\end{figure}
Assume we place a guard at a random position inside this intersection.
Then {none} of the cones $c_j$ is empty anymore.
That means that there are new deepest cones with different slopes,
since they have to be tangent to the new guard
(see Figure~\ref{fig:squared6}, middle).
Now we pull the new guard in the direction of the medial axis
of the former deepest cones towards infinity.
Then we observe, that while we move this point,
the deepest cones are rotated and
this way are pulled away from the $x$-axis.
Since the limes of the rotation,
compared to the slope of the original deepest cones $c_j$,
has absolute value  $\frac{\Theta_i}{2}$,
we can force a rotation by an angle
that is arbitrarily close to half of the apex angle,
i.e.\ $\frac{\Theta_i}{2}-\varepsilon$ for any $\varepsilon>0$
(see Figure~\ref{fig:squared6}, right).
Note the generality of the above discussion for all $h=\{0,\dots,2i-1\}.$
We will definitely \emph{not} place a guard in the \emph{union}
of the deepest cones for $h=0,$
since these tunnels force the decomposition of the $\Theta_i$-region
in the center.
In spite of this we have used exactly this case
in the drawings in Figure~\ref{fig:squared6},
since it depicts the worst-case scenario we will consider later in the proof.

Now we are able to complete the construction.
First we compute the slope of the medial axes
of the deepest cones $c_j$ as well as $d_j$ of Observation~\ref{obs:deepest}
for all $h=\{1,\dots,2i-1\}.$
For each of these $4i-2$ slopes
we place a guard at the intersection point
of the ray emanating from the origin having this slope
and the boundary of a new box $B_x$
(see Figure~\ref{fig:squared10}, left).
\begin{figure}[t]
  \centering
  \includegraphics[width=.7\columnwidth]{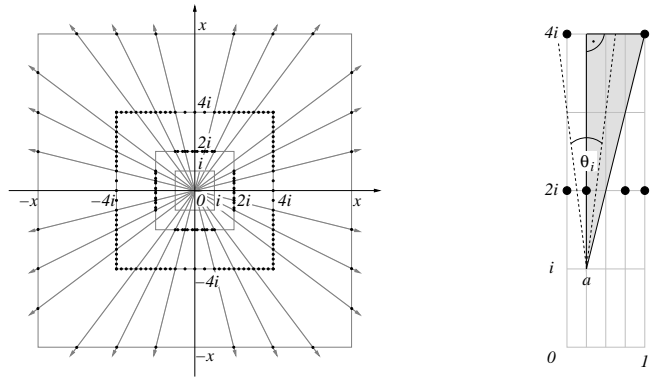}
  \caption{Construction of the second step for all parts:
     above, below, left, and right (left).
     Construction used in the proof of Lemma~\ref{lem:unwanted} (right).}
  \label{fig:squared10}
\end{figure}
Box $B_x$ has \emph{necessarily} to be large enough to guarantee,
that each intersection point lies
\begin{enumerate}
 \item outside the union of the wanted cones and
 \item inside the intersection of the deepest cones of the given slope.
\end{enumerate}
The existence of box $B_x$ with these properties
follows from the discussion above.
But this is not \emph{sufficient}.
It remains to prove the following lemma.
\begin{lemma}\label{lem:unwanted}
Box $B_x$ can be chosen large enough such that
no empty $\Theta_i$-cone, but the wanted cones, can intersect $B_i.$
\end{lemma}
\begin{proof}
It is sufficient to prove the claim for the deepest cones 
with the minimum apex $y$-value amongst all deepest cones
according to Observation~\ref{obs:deepest}.
These are the cones for $h=0.$
(Please note that we do not block tunnels for $h=0$ in practice;
we just prove that we could even hinder cones through these tunnels
from entering $B_i.$)

Remember that we can place the guard which blocks the deepest cones
such that the deepest cones are rotated by an angle
arbitrarily close to $\frac{\Theta_i}{2}.$
W.l.o.g. we assume that the cone is rotated clockwise.
Consider the empty cone $c$ of maximum angle with apex at $a=(\frac{1}{4},i)$
inside a cell with guard pattern A
(shaded region in Figure~\ref{fig:squared10}, right).
This cone touches the boundary of $B_i$ and its
left ray is vertical as it is the case for maximal rotated deepest cones.
If we can show that the angle of $c$ is smaller than $\Theta_i$
it follows that $c$ can not enter $B_i$.
We crop $c$ at the line $y=4i$ to make it a rectangular triangle.
Now we take the $\Theta_i$-cone from pattern A, Figure~\ref{fig:squared3},
and move its apex to $a.$
We divide the triangle along the right boundary of the $\Theta_i$-cone
through point $(\frac{5}{8},4i).$
Consequently the left sub-triangle has angle $\frac{\Theta_i}{2}$
at point $a.$
Since the opposite leg of the entire triangle is 2-times the opposite leg
of the left sub-triangle,
the total angle of $c$ at $a$ is less than 2-times $\frac{\Theta_i}{2}.$
\qed
\end{proof}

After repeating this construction
for the lower, left, and right half
we have placed $16i-8$ additional guards.
Together with the guards from the first step they define the set $G_i$
with $n_i = 96i-4$ guards in total.
The generic example presented in this section proves the following Theorem.
\begin{theorem}
There is a sequence of inputs $(\Theta_i,n_i,G_i)_{i\in\N}$
with $\lim_{i\to\infty}\Theta_i=0$
such that the asymptotic bound on the
complexity of their $\Theta$-region is $\Omega(n^2)$
where $n$ is the number of guards.
\end{theorem}

\section{Algorithm}\label{s:algorithm}

Here we discuss a way to compute the boundary of the $\Theta$-region.
Note that we gave bounds on the worst-case complexity
of the $\Theta$-region above. 
Clearly, for any $n$ and any $\Theta$ there are sets $G$
for which the $\Theta$-region is empty or extremely simple.
Despite of this our algorithm will consider the $O(\frac{n}{\Theta})$
arcs in $\mathcal{C}$ and hence can not be output-sensitive.
We allow a simplification in the presentation of the algorithm:
We will consider a set $\mathcal{C'}$ of arcs
which are longer on one side, i.e.\ $|\mathcal{C}|=|\mathcal{C'}|$ and
$\bigcup \mathcal{C} \subset \bigcup \mathcal{C'}$.

First we compute the convex hull $\ch(G)$ and
add for each hull edge $(u,v)$
the circular arc $C_{u,v}^{\Theta}$ to the set $\mathcal{C'}$.
For each guard $g$, that is not a vertex of $\ch(G)$,
we compute all empty cones of maximal angle with apex at $g$
together with two guards (witnesses) $g_{\min}$ and $g_{\max}$
per empty cone, which lie on its rays.
(See the light shaded cone in Figure~\ref{fig:algo2}.)
This can be done by
using the algorithm of Avis et al.~\cite{avis}
in $O((\frac{n}{\Theta})\log n)$ time and $O(n)$ space. 
\begin{figure}[t]
  \centering
  \includegraphics[width=.25\columnwidth]{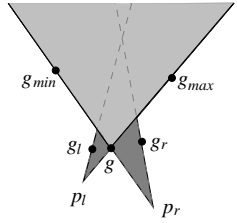}
  \caption{Guard $g$ can contribute to at most two end points
	   for each maximal empty cone with apex $g.$}
  \label{fig:algo2}
\end{figure}

As we did in the proof of Theorem~\ref{tm:size},
we find the arcs in $\mathcal{C}$ via their end points.
If we move an empty $\Theta$-cone with apex $g$
and its left ray through $g_{\min}$
along the line through $g$ and $g_{\min}$
until a guard, say $g_{r}$, is tangent to the other ray,
the new apex marks an end point $p_r$ of two arcs in the set $\mathcal{C}$
(see Figure~\ref{fig:algo2}).
Since we do not know
the second end points of the arcs,
we add the piece of $C_{g_{\min},g_r}^\Theta$ to $\mathcal{C'}$
that ends in $g_{r}$ and ${p_r}$,
and we add the piece of $C_{g,g_r}^\Theta$ to $\mathcal{C'}$
that ends in $g$ and ${p_r}$,
A similar construction for the line through $g$ and $g_{\max}$
will add another two arcs to $\mathcal{C'}$.

Note that by   
fixing the line through $gg_{\min}$ we can find 
guard $g_r$
naively by simply inspecting all guards in $G$, 
and similarly we can find $g_l$ for the line through $gg_{\max}$.
However, 
one can compute guards $g_r$ and $g_l$ faster with the help of the well-know 
\emph{Partition Theorem} that 
has been extensively used in the context of  range searching. 
We cite the theorem for a planar point set.
\begin{theorem}(Partition Theorem \cite{Matousek}.)\label{TM:Partition}
Any set $S$ of $n$ points in the plain
can be partitioned into $O(r)$ disjoint classes by a simplicial partition,
such that every simplex (i.e. triangle)
contains between $\frac{n}{r}$ and $\frac{2n}{r}$ points and
every line crosses 
at most $O(r^{\frac{1}{2}})$ simplices (crossing number). Moreover, for any
$\xi > 0$ such a 
simplicial partition can be constructed in $O(n^{1+\xi})$ time.
\end{theorem}
Using this Theorem recursively one can construct a tree which is
called a \emph{partition tree} (e.g. the 
root of the tree, associated with $S$, has $O(r)$ children, 
each associated with a simplex from the first level, and so on). From now on we assume that 
$r$ is a constant. 
Observe that if $r$ is a constant, the partition tree is of $O(n)$ size
and it can be constructed in $O(n^{1+\xi})$ time for any $\xi>0$. 
\begin{lemma}
For any $\xi>0$, there is a data structure of $O(n\log n)$ size and $O(n^{1+\xi})$ construction time
such that for the given lines through 
$gg_{\min}$ and $gg_{\max}$, corresponding guards $g_l$ and $g_r$ can be computed 
in additional $O(n^{\frac{1}{2}+\xi})$ time.  
\end{lemma}
\begin{proof}
  Assume we are given a partition tree and suppose that we fix 
  the line through $gg_{\max}$.  Clearly by Theorem~\ref{TM:Partition} we
  have the bound on the number of triangles that intersect the line
  which is $O(\sqrt{r})$. On those triangles we recur,
  which leads to a total of $O(\sqrt{n})$ triangles intersected by the
  line. But still there might be $O(r)$ triangles lying
  completely to the left of the line $gg_{\max}$. 
  For those triangles
  we can precompute a convex hull for the points inside each
  triangle. This will
  increase the total space of the partition tree by a $O(\log n)$
  factor since every level in the tree now will be of $O(n)$ size.
  However, this way we avoid recursing on the triangles that lie
  completely to the left of $gg_{\max}$. Namely, for every triangle that lies
  to the left of $gg_{\max}$, guard $g_l$ can be found as an 
  extreme point of the precomputed convex hull in the direction perpendicular to the line
  that forms the $\Theta$-cone with the line through $gg_{\max}$ in $O(\log n)$ total time
  (see~\cite{orourke00}, Section 7.9).
  The case for the line through $gg_{\min}$ is similar.
  \qed
\end{proof}

Therefore we can state the following Lemma.
 
\begin{lemma}
For any $\xi>0$, the set $\mathcal{C'}$ can be computed in $O(n^{\frac{3}{2}+\xi}/\Theta)$ time and $O(n\log n)$ space. 
\end{lemma}

Next we discuss how to compute the $\Theta$-guarded region from the set of circular arcs $\mathcal{C'}$. 
For each connected component of the $\Theta$-region the algorithm  
outputs a sequence  $p_1,\ldots, p_k$ of points in the plane and circular 
arcs incident with pairs $p_{i-1}, p_i$ for
$i=2,\ldots, k$ and $p_k, p_1$ as edges of the $\Theta$-region.

We start with computing the {arrangement} 
$\mathcal{A}(\mathcal{C'})$ of set $\mathcal{C'}$. 
Let $\psi$ denote the number of cells in $\mathcal{A}(\mathcal{C'})$ and
let $\mu$ denote the total complexity of
the arrangement $\mathcal{A}(\mathcal{C'})$,
which upper bounds the complexity of the $\Theta$-region. 
Edelsbrunner et al.~\cite{edelsbrunner92} showed that $\mu$ is at most
$O(\sqrt{\psi} (\frac{n}{\Theta}) 2^{\alpha(n)})$, where $\alpha(\cdot)$ is the inverse 
Ackerman function which is an extremely slow-growing function. 
Moreover, the arrangement $\mathcal{A}(\mathcal{C'})$
can be constructed 
in $O((n+\mu) \log n)$ time by the plane-sweep algorithm
of Bentley and Ottman \cite{bo79}

Since arcs in $\mathcal{C'}$ are bounding circular segments from the Formula~(\ref{f:complexity}), 
cells in the arrangement $\mathcal{A}(\mathcal{C'})$ will have the property that  they 
are either $\Theta$-guarded or not $\Theta$-guarded. 
Hence, if some point from the 
cell is $\Theta$-guarded then the whole cell belongs to the $\Theta$-region and opposite. 
Let $P$ denote the set of $\psi$ different points such that each point is taken from the interior 
of $\psi$ different cells in $\mathcal{A}(\mathcal{C'})$. To detect the cells that belong
to the $\Theta$-region, we use the following lemma. 
\begin{lemma}[Avis et al.~\cite{avis}] \label{lm:avis}
Let $G$ be a set of $n$ guards and
let $P$ be a set of $\psi$ query points in $\R^2$.
The $\Theta$-unguarded points 
of $P$ can be reported together with their witnesses $g_{\min}$ and $g_{\max}$
in $O(\frac{n+\psi}{\Theta}\log (n+\psi))$ time and $O(n)$ space. 
\end{lemma}
\begin{proof}
Avis et al. \cite{avis} presented an algorithm to compute all
$\Theta$-unguarded guards of $G$ in $O(\frac{n}{\Theta}\log n)$ time 
and $O(n)$ space.
So far the set of query points and the set of guards are the same.
But since
their algorithm actually distinguishes between query points and guards,
it can be extended immediately:
In Steps~2 and 3 of procedure \emph{Unoriented Maxima} on page 284f.
only guards are inserted into the convex hull constructions,
while tangents to these convex hulls
are only computed through query points.
The running time of the algorithm is dominated by sorting 
the points in $G\cup P$ for
$\frac{\pi}{\Theta}$ %
many directions
which takes $O(\frac{n+\psi}{\Theta}\log (n+\psi))$ time.
For more details see Section 2 and Section 6 (Appendix) in~\cite{avis}.
\qed
\end{proof}

At the end we collect all $\Theta$-guarded
cells and output the sequence of nodes and edges on the 
boundary of the union of them. 
We conclude with the following Theorem. 
\begin{theorem} 
For any $\xi>0$, the $\Theta$-region for $\Theta < \pi$ can be computed in time  
$O(n^{\frac{3}{2}+\xi}/\Theta + \mu \log n)$,
where $\mu$ denotes the complexity of the arrangement $\mathcal{A}(\mathcal{C'})$.
\end{theorem}

\section{Conclusion}\label{s:conclude}
In this paper we consider a point to be guarded,
if it is guarded from 'all' sides
by a given finite set of guards $G$.
Our main goals were to analyze the shape and
the complexity of the $\Theta$-region,
i.e.\ the set of all $\Theta$-guarded points,
and give a mathematical description of it.

As a result,
we showed that the $\Theta$-region is defined by a set of at most
$O(\frac{n}{\Theta})$ many circular arcs. 
The difficulty in the complexity analysis of the $\Theta$-region itself
appeared while arguing about the complexity of the union of convex sets $U_i$
which are bounded by these arcs (cf. Formula~\ref{f:complexity}).
In dependency on $\Theta$
we summarize our results on the worst-case complexity
of the $\Theta$-region in Table~1.
Furthermore, we could give a series of inputs
with decreasing angle and increasing number of guards
whose asymptotic complexity is $\Omega(n^2)$.
Finally we gave an algorithm to compute the $\Theta$-region.
\begin{table}[t]%
\begin{center}\begin{tabular}{|@{\quad}c@{\quad}|@{\quad}c@{\quad}|}
\hline
\rule{0ex}{2.5ex} angle $\Theta$              & worst-case complexity \\[.5ex]\hline
\rule{0ex}{2.5ex} $\pi\le\Theta<2\pi$           & $|\ch(G)|$ vertices \\[.5ex]
$\frac{\pi}{2}\le\Theta<\pi$ & $O(n)$ \\[.5ex]
$ \delta <\Theta<\frac{\pi}{2},$ for constant $\delta>0$ &  $O(n^{1+\varepsilon}),$ for any $\varepsilon>0$ \\[.5ex]\hline
\end{tabular}\end{center}%
 \caption{The worst-case complexity of the $\Theta$-region
      in dependency on the angle $\Theta$.}%
\end{table}

\section{Appendix: Technical Mistakes in a Recent Publication}
Abellanas et al.~\cite{ACM08} claim the complexity
of the $\Theta$-region to be $O(n)$ for a fixed value of $\Theta\in(0,\pi]$
and suggest an algorithm to compute it.
Because the paper is not precise in several places,
we will point out mistakes
that show gaps in the proofs of the main results.

\subsubsection*{Mistake in the Construction.}
The entire construction, and hence the results,
are based on the existence of certain arc chains,
one per convex hull edge,
which are introduced in Definition~4 on page~367.
But Definition~4 is actually a characterization of a point-set.
That means Definition~4 contains a \emph{hidden and unproved Lemma,}
stating that the defined point-set forms a chain of arcs.
In the context of other mistakes
we will argue below
that these chains \emph{do not exist in general for acute $\Theta<\frac{\pi}{2}.$}

We begin with Definition~3 on page~366.
(Be aware of the index typos.)
It is stated that a guard pair $(s_i,s_j)$ can only be
entered by empty cones from one direction,
and hence one of the two connecting arcs is excluded a priori.
But this is not generally true for $\Theta<\frac{\pi}{2}$
as the guard pair $(a,c)$ in Figure~\ref{fig:mistakes_4} shows.
In the following we assume that, however,
both directions
can be distinguish and handled somehow.

Next we focus on the construction of the chain
as is given in the proofs of Lemma~3 and Lemma~4 on pages~368--370.
There vertices are connected via two arc pieces to a predecessor
and a successor only, guaranteeing a vertex degree of 2 by construction.

Because of the mentioned mistake in Definition~3,
there is now a misunderstanding of the terms
predecessor and successor of an arc end point:
The predecessor (resp. successor) is the left (right) end point
of the connecting arc piece
\emph{viewed from the direction of the empty cone}
that supports this arc.
(By the way, this cannot be decided only by the $x$-coordinate of guards,
as is assumed for the partitioning of the set $M_{s_k,s_{k+1}}$
in the left column, above Figure~3, on page 367.)

\begin{figure}[t]
  \centering
  \includegraphics[width=.6\columnwidth]{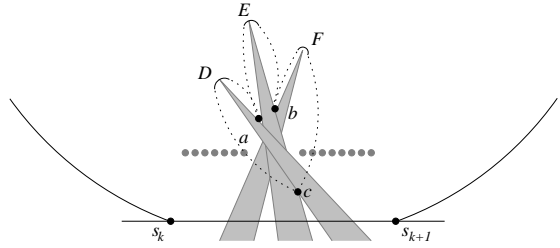}
  \caption{Example that guards can not be ordered linearly
    for $\Theta<\frac{\pi}{2}$ in all cases.}
  \label{fig:mistakes_4}
\end{figure}
Now we have a look at \emph{our} Figure~\ref{fig:mistakes_4}.
There $(s_k,s_{k+1})$ is a convex hull edge and
$a,$ $b,$ $c$ denote guards.
There may be also additional (shaded) guards but we rely on the existence
of the three empty (shaded) cones.
Then $D,$ $E,$ and $F$ denote non-empty
arc pieces on the pairwise arc chains (dotted lines) between these guards.
During the construction the guards (even if they are removed later on)
and arc pieces get their position in the chain
that belongs to the convex hull edge.
We derive an order on these objects from the empty cones:
$c\prec D\prec a$ and $a\prec E\prec b$ and $b\prec F\prec c.$
These constraints on the order can only be satisfied
by a loop that contains these guards and arc pieces,
since every vertex has degree 2.
This contradicts the assumption of the existence of an arc chain.

\paragraph{Remark:}
Because of this general argumentation,
we are convinced that
the task of finding the $\Theta$-region has to be motivated,
however, rather by the directions of empty cones
than by the convex hull edges.
(Note that each instance of the generic example in Section~\ref{s:lowerbound}
can be embedded in a huge box with just four hull edges.)

\subsubsection*{Mistake in the Complexity Statement.}
In Proposition~2 on page~370
the complexity of the $\Theta$-region is claimed to be
linear for any fixed value of $\Theta.$
Below we point out that
this complexity bound is not justified
in the {most crucial} step of the proof,
hence \emph{leaving a gap in the proof.}

The proof begins with a consideration on
the boundary of the $\Theta$-region
restricted to a convex hull edge ${(s_k,s_{k+1})},$
i.e.\ cones can now only enter the hull through ${(s_k,s_{k+1})}.$
It is stated,
that the boundary is the chain of arcs that corresponds to this edge
and that it has linear size.
Although we have proven above
that this chains do not exist in general for $\Theta<\frac{\pi}{2},$
and hence we do not know anything about the complexity of
a proper replacement for these chains,
\emph{we assume just for the sake of argumentation} that this would be true, however.

In the second half of the proof
the total number of vertices in all arc chains is estimated.
But what is missing in the counting argument for the complexity of
the union of the $\Theta$-regions, that are restricted to a convex hull edge,
are the
\emph{intersections between these restricted $\Theta$-regions.}
(Remember that we have built our examples with quadratic complexity {only}
on this kind of intersections in Section~\ref{s:lowerbound}.)
No characteristic of the objects is given here,
why the complexity of the union should be $O(n)$ for fixed values of $\Theta;$
even fatness and convexity seem not to be strong enough for this claim.

The whole proof of this Proposition
seems to go much more along the line of proving 
that the $\Theta$-region can be \emph{described} by a linear number of arcs,
which we have also proven in Theorem~\ref{tm:size}, than 
arguing about the complexity of the $\Theta$-region itself. 
Without further argumentation there is just
an upper bound of
$O(n^2)$ on the complexity for constant angles
since the $\Theta$-region is embedded
in the arrangement of $O(n)$ arcs.

\small 
\bibliographystyle{abbrv}

\end{document}